\documentclass[10pt,letterpaper]{article}
\usepackage[ruled]{algorithm2e}
\usepackage{algorithmic}
\usepackage{authblk}
\usepackage{multirow}
\usepackage{amsmath}
\usepackage{amssymb}
\usepackage{amsthm}
\usepackage[T1]{fontenc}
\usepackage{times}
\usepackage[backend=bibtex,style=alphabetic,maxnames=8,maxalphanames=3,maxcitenames=3,minalphanames=3,minnames=1]{biblatex}
\usepackage{geometry}
\newcommand{\call}{\mathcal{C}_\mathrm{all}}
\newcommand{\cstart}{\mathcal{C}_\mathrm{start}}
\newcommand{\ccolor}{\mathcal{C}_\mathrm{color}}
\newcommand{\bstart}{\mathcal{B}_\mathrm{start}}
\newcommand{\ie}{{\em i.e.,~}}
\newcommand{\eg}{{\em e.g.,~}}

\newcommand{\plog}{P_{LL}}

\newcommand{\cmax}{c_{\mathrm{max}}}
\newcommand{\lmax}{l_{\mathrm{max}}}
\newcommand{\smax}{s_{\mathrm{max}}}
\newcommand{\rmax}{r_{\mathrm{max}}}
\newcommand{\rsize}{\Phi}
\newcommand{\tr}{\mathit{true}}
\newcommand{\fl}{\mathit{false}}

\newcommand{\cand}{A}
\newcommand{\timer}{B}
\newcommand{\initial}{X}
\newcommand{\vtimer}{V_{\timer}}
\newcommand{\vcand}{V_{\cand}}
\newcommand{\vl}{V_{L}}
\newcommand{\vf}{V_{F}}
\newcommand{\vinitial}{V_{\initial}}

\newcommand{\var}{\mathtt{var}}
\newcommand{\leader}{\mathtt{leader}}
\newcommand{\level}{\mathtt{level}}
\newcommand{\levelq}{\level_Q}
\newcommand{\levelb}{\level_B}
\newcommand{\cnt}{\mathtt{count}}
\newcommand{\rand}{\mathtt{rand}}
\newcommand{\ind}{\mathtt{index}}
\newcommand{\epoch}{\mathtt{epoch}}
\newcommand{\done}{\mathtt{done}}
\newcommand{\init}{\mathtt{init}}
\newcommand{\status}{\mathtt{status}}
\newcommand{\tick}{\mathtt{tick}}
\newcommand{\clr}{\mathtt{color}}

\newcommand{\clog}[1]{\lceil \log_2 #1 \rceil}
\newcommand{\sinit}{s_{\mathrm{init}}}
\newcommand{\outputs}{\pi_{\mathrm{out}}}
\newcommand{\cinit}[1]{C_{\mathrm{init},#1}}

\newcommand{\countup}{\mathit{CountUp}}
\newcommand{\quick}{\mathit{QuickElimination}}
\newcommand{\tourn}{\mathit{Tournament}}
\newcommand{\backup}{\mathit{BackUp}}

\newcommand{\ifthen}[2]{\textbf{if} #1 \textbf{then} #2 \textbf{endif}}
\newcommand{\foralldo}[2]{\textbf{for all} #1 \textbf{do} #2 \textbf{endfor}}

\newtheorem{theorem}{Theorem}
\newtheorem{lemma}{Lemma}
\theoremstyle{definition}
\newtheorem{definition}{Definition}
\newcommand{\ex}{\mathbf{E}}

\newcommand{\unit}{\left \lfloor 21 n \ln n \right \rfloor}
\newcommand{\smallunit}{\left \lfloor 9 n \ln n \right \rfloor}

\newcommand{\loglogn}{\left \lceil \lg \lg n \right \rceil}

\title{Logarithmic Expected-Time Leader Election in Population Protocol Model}
\date{}
\author[1]{Yuichi Sudo\thanks{Corresponding author:y-sudou[at]ist.osaka-u.ac.jp}}
\author[2]{Fukuhito Ooshita}
\author[3]{Taisuke Izumi}
\author[4]{Hirotsugu Kakugawa}
\author[1]{Toshimitsu Masuzawa}

\affil[1]{Graduate School of Information Science and Technology, Osaka University, Japan}
\affil[2]{Graduate School of Science and Technology, Nara Institute of Science and Technology, Japan}
\affil[3]{Graduate School of Engineering, Nagoya Institute of Technology, Japan}
\affil[4]{Faculty of Science and Technology, Ryukoku University, Japan}

\begin{document}
\maketitle
\begin{abstract}
\normalsize
In this paper, we present the \emph{first} leader election protocol in the population protocol model that
stabilizes $O(\log n)$ parallel time
in expectation with $O(\log n)$ states per agent,
where $n$ is the number of agents.
Given a rough knowledge $m$ of the population size $n$
such that $m \ge \log_2 n$ and $m=O(\log n)$,
the proposed protocol guarantees that
exactly one leader is elected and the unique leader is kept forever thereafter.
\end{abstract}

\section{Introduction}
\label{sec:intro}

We consider the \emph{population protocol} (PP) model \cite{original} in
 this paper.
A network called \emph{population} consists of a large number of
 finite-state automata,
called \emph{agents}.
Agents make \emph{interactions}
 (i.e., pairwise communication) each other
by which they update their states.
The interactions are opportunistic,
that is, they are unpredictable.
Agents are strongly anonymous:
they do not have identifiers
and they cannot distinguish
their neighbors with the same states.
As with the majority of studies on population protocols,
we assume that the network of agents is a complete graph
and that the scheduler
selects an interacting pair of agents at each step
uniformly at random.

In this paper, we focus on leader election problem, which is one of the most fundamental
and well studied problems in the PP model.
The leader election problem
requires that starting from a specific initial
configuration,
a population reaches a safe configuration in which exactly one leader exists
and the population keeps that unique leader
thereafter.

There have been many works which study the leader election problem in the PP model (Tables \ref{tbl:protocols} and 2). 
Angluin et al.~\cite{original}
gave the first leader election protocol,
which stabilizes in $O(n)$ parallel time in expectation
and uses
only constant space of each agent,
where $n$ is the number of agents
and ``parallel time'' means
the number of steps in an execution divided by $n$.
If we stick to constant space,
this linear parallel time is optimal;
Doty and Soloveichik \cite{DS18} showed that
any constant space protocol requires linear parallel time
to elect a unique leader.
Alistarh and Gelashvili \cite{AG15} made a breakthrough
in 2015;
they achieve poly-logarithmic stabilization time
($O(\log^3 n)$ parallel time)
by increasing the number of states from $O(1)$
to only $O(\log^3 n)$.
Thereafter, the stabilization time has been improved
by many studies \cite{BCER17,AAG18,GS18,GSU18,MST18}.
G{\k{a}}sieniec et al.~\cite{GSU18} gave
a state-of-art protocol that stabilizes in
$O(\log n \cdot \log \log n)$ parallel time
with only $O(\log \log n)$ states.
Its space complexity is optimal;
Alistarh et al.~\cite{AAE+17} shows that
any leader election algorithm
with $o(n/(\mathrm{polylog}\ n))$ parallel time
requires $\Omega(\log \log n)$ states.
Michail et al.~\cite{MST18} gave a protocol
with $O(\log n)$ parallel time 
but with a linear number of states.
Those protocols with non-constant
number of states~\cite{AG15,AAE+17,BCER17,AAG18,GS18,GSU18} are not \emph{uniform}, that is, they require
some rough knowledge of $n$.
For example,
in the protocol of \cite{GS18},
an $\Theta(\log \log n)$ value
must be hard-coded to set the maximum value of
one variable (named $l$ in that paper).
One can find detailed information about the leader election in the PP model in two survey papers \cite{AG18,ER18}.

The stabilization time of \cite{MST18} is optimal;
any leader election algorithm requires $\Omega(\log n)$ parallel time if it uses any large number of states and assumes the exact knowledge of population size $n$ \cite{SM19}.
At the beginning of an execution,
all the agents are in the same initial state
specified by a protocol.
Therefore, simple analysis on Coupon Collector's problem
shows that we cannot achieve $o(\log n)$ parallel stabilization time if an agent in the initial state is a leader.
The lower bound of \cite{SM19} shows that
we cannot achieve $o(\log n)$ parallel time
even if we define the initial state such that
all the agents are non-leaders initially.


\paragraph{Our Contribution}
In this paper, we present the \emph{first}
time-optimal leader election protocol $\plog$
with sub-polynomial number of states.
Specifically, the proposed protocol $\plog$ stabilizes
in $O(\log n)$ parallel time and uses only
$O(\log n)$ states per agent.
%
Compared to the state or art protocol \cite{GSU18},
$\plog$ achieves shorter
(and best possible)
stabilization time but uses larger space of each agent.
Compared to \cite{MST18},
$\plog$ achieves drastically small space
while maintaining the same (and optimal) stabilization time.
The protocol $\plog$ is non-uniform
as with the existing non-constant space protocols;
it requires
a rough knowledge $m$ of $n$ such that $m \ge \log_2 n$
and $m=\Theta(\log n)$.

We give $\plog$ as an asymmetric protocol in the main part of this paper \emph{only for simplicity}
of presentation and analysis of stabilization time.
Actually, we can change $\plog$ to a symmetric protocol,
which we discuss in Section \ref{sec:symmetric}.
In particular,
that section proposes
the first implementation of totally independent and fair (\ie unbiased) coin flips in the symmetric version
of the PP model.
Although the implementation of coin flips
in \cite{AAE+17} is almost independent and fair, 
the totally independent and fair coin clips
achieved in this paper can contribute
a simple analysis in a variety kind of
protocols in the PP model.


\begin{table}
\center
\caption{Leader Election Protocols. (Stabilization time is shown in terms of parallel time and in expectation.)}
\label{tbl:protocols}
\vspace{0.3cm}

\begin{tabular}{c c c}
\hline
 & States & Stabilization Time\\
\hline
\cite{original}
& $O(1)$ & $O(n)$\\
\cite{AG15}
&$O(\log^3 n)$ & $O(\log^3 n)$\\
\cite{AAE+17}& $O(\log^2 n)$
&$O(\log^{5.3}n \cdot \log \log n)$\\
\cite{AAG18}&
$O(\log n)$&$O(\log^2 n)$\\
\cite{GS18}&
$O(\log \log n)$&$O(\log^2 n)$\\
\cite{GSU18}&
$O(\log \log n)$&$O(\log n \cdot \log \log n)$\\
\cite{MST18}& $O(n)$&$O(\log n)$\\
This work&
$O(\log n)$&$O(\log n)$\\
\hline
\end{tabular}
\end{table}

\begin{table}
\center
\caption{Lower Bounds for Leader Election (Stabilization time is shown in terms of parallel time and in expectation.)}
\label{tbl:lower}
\vspace{0.3cm}

\begin{tabular}{c c c}
\hline
 & States & Stabilization Time\\
\hline
\cite{DS18} & $O(1)$ & $\Omega(n)$\\
\cite{AAE+17} & $<1/2 \log \log n$ & $\Omega(n/ (\mathrm{polylog}\ n))$\\
\cite{SM19} & any large  & $\Omega(\log n)$\\
\hline
\end{tabular}
\end{table}

\section{Preliminaries}
\label{sec:pre}

A \emph{population} is
a network consisting of {\em agents}.
We denote the set of all the agents by $V$ and let $n = |V|$.
We assume that a population is complete graph,
thus every pair of agents $(u,v)$ can interact,
where $u$ serves as the \emph{initiator}
and $v$ serves as the \emph{responder} of the interaction.
Throughout this paper, we use the phrase
``with high probability'' to denote
probability $1-O(n^{-1})$.

A \emph{protocol} $P(Q,\sinit,T,Y,\outputs)$ consists of 
a finite set $Q$ of states,
an initial state $\sinit \in Q$,
a transition function
$T:  Q \times Q \to Q \times Q$,
a finite set $Y$ of output symbols, 
and an output function $\outputs : Q \to Y$.
Every agent is in state $\sinit$
when an execution of protocol $P$ begins.
When two agents interact,
$T$ determines their next states
according to their current states.
The \emph{output} of an agent is determined by $\outputs$:
the output of an agent in state $q$ is $\outputs(q)$.
In this paper, we assume that a rough knowledge
of an upper bound of $n$ is available.
Specifically, we assume that an integer
$m$ such that $m \ge \log_2 n$
and $m=\Theta(\log n)$ are given,
thus we can design $P(Q,\sinit,T,Y,\outputs)$
using this input $m$,
\ie the parameters $Q$,$\sinit$,$T$,$Y$,
and $\outputs$ can depend on $m$.

A \emph{configuration} is a mapping $C : V \to Q$ that specifies
the states of all the agents.
We define $\cinit{P}$ as the configuration of $P$
where every agent is in state $\sinit$.
We say that a configuration $C$ changes to $C'$ by the interaction
$e = (u,v)$,
denoted by $C \stackrel{e}{\to} C'$,
if
$(C'(u),C'(v))=T(C(u),C(v))$
and $C'(w) = C(w)$
for all $w \in V \setminus \{u,v\}$.

A \emph{schedule} $\gamma = \gamma_0,\gamma_1,\dots
=(u_0,v_0),(u_1,v_1),\dots~$
is a sequence of interactions.
A schedule determines which interaction occurs at each \emph{step},
{\ie} interaction $\gamma_t$ happens at step $t$
under schedule $\gamma$.
In particular, we consider a \emph{uniformly random scheduler}
$\Gamma=\Gamma_0, \Gamma_1,\dots$ in this paper:
each $\Gamma_t$ of the infinite sequence of interactions
is a random variable such that
$\Pr(\Gamma_t = (u,v)) =\frac{1}{n(n-1)}$
for any $t \ge 0$ and any distinct $u,v \in V$.
Note that we use capital letter $\Gamma$
for this uniform random scheduler
while we refer a deterministic schedule
with a lower case such as $\gamma$.
Given an initial configuration $C_0$ and
a schedule $\gamma$,
the \emph{execution} of protocol $P$ is uniquely defined as 
$\Xi_{P}(C_0,\gamma) = C_0,C_1,\dots$ such that
$C_t \stackrel{\gamma_t}{\to} C_{t+1}$ for all $t \ge 0$. 
Note that the execution $\Xi_{P}(C_0,\Gamma) = C_0,C_1,\dots$
under the uniformly random scheduler $\Gamma$
is a sequence of configurations where
each $C_i$ is a random variable.
For a schedule $\gamma = \gamma_0,\gamma_1,\dots$
and any $t \ge 0$,
we say that agent $v \in V$ \emph{participates} in $\gamma_t$
if $v$ is either the initiator or the responder of $\gamma_t$.
We say that a configuration $C$ of protocol $P$ is reachable 
if there exists a finite schedule $\gamma=\gamma_0,\gamma_1,\dots,\gamma_{t-1}$ such that
$\Xi_P(\cinit{P},\gamma)=C_0,C_1,\dots,C_{t}$ and $C=C_{t}$.
We define $\call(P)$ as the set of all reachable configurations of $P$.

The leader election problem requires that 
every agent should output $L$ or $F$ which means
``leader'' or ``follower'' respectively.
Let $\cal{S}_P$ be the set of configurations
such that, for any configuration $C \in \cal{S}_P$,
exactly one agent outputs $L$ (\ie is a leader) in $C$
and no agent changes its output in
execution $\Xi_P(C,\gamma)$ for any schedule $\gamma$.
We say that a protocol $P$
is a leader election protocol or
solves the leader election problem if
execution $\Xi_P(\cinit{P},\Gamma)$
reaches a configuration in $\cal{S}_P$ with probability $1$.
For any leader election protocol $P$, 
we define the expected stabilization time of $P$
as the expected number of steps during which
execution $\Xi_P(\cinit{P},\Gamma)$
reaches a configuration in $\cal{S}_P$,
divided by the number of agents $n$.
The division by $n$ is needed
because we evaluate the stabilization time
in terms of parallel time.

We write the natural logarithm of $x$ as $\ln x$
and the logarithm of $x$ with base $2$ as $\lg x$.
We do not indicate the base of logarithm
in an asymptotical expression such as $O(\log n)$.
By an abuse of notation, we will identify an interaction $(u,v)$ with the set $\{u,v\}$ whenever convenient.

Throughout this paper, we will use
the following three variants of Chernoff bounds. 

\begin{lemma}[\cite{kyoukasyo}, Theorems 4.4, 4.5]
\label{lemma:chernoff}
Let $X_1,\dots,X_s$ be independent Poisson trials,
and let $X = \sum_{i=1}^s X_i$.
Then
\begin{align}
 \label{eq:upperdouble}
 \forall \delta,~0 \le \delta \le 1:~
 \Pr(X \ge (1+\delta)\ex[X]) &\le e^{-\delta^2\ex[X]/3},\\
 \label{eq:lowerhalf}
 \forall \delta,~0 < \delta < 1:~
 \Pr(X \le (1-\delta)\ex[X]) &\le e^{-\delta^2\ex[X]/2}.
\end{align}
\end{lemma}

In the proposed protocol,
we often use \emph{one-way epidemic} \cite{fast}.
The notion of one-way epidemic is formalized as follows.
Let $\gamma=\gamma_0,\gamma_1,\dots$
be an infinite sequence of interactions,
$V'$ be a set of agents ($V'\subseteq V$),
and $r$ be an agent in $V'$.
The \emph{epidemic function}
$I_{V',r,\gamma}:[0,\infty) \rightarrow 2^V$
is defined as follows:
$I_{V',r,\gamma}(0) = \{r\}$,
and for $t = 1,2,\dots$,
$I_{V',r,\gamma}(t) = I_{V',r,\gamma}(t-1) \cup (\gamma_{t-1}\cap V')$
if $I_{V',r,\gamma}(t-1) \cap \gamma_{t-1} \neq \emptyset$;
otherwise, $I_{V',r,\gamma}(t) = I_{V',r,\gamma}(t-1)$.
We say that
$v$ is \emph{infected\/} at step $t$
if $v \in I_{V',r,\gamma}(t)$ 
in the epidemic in $V'$ and under $\gamma$ starting from agent $r$.
At step 0, only $r$ is infected; at later steps, an agent in $V'$
becomes infected if it interacts with an infected agent.
Once an agent becomes infected, it remains infected thereafter.

This abstract notion plays an important role in analyzing
the expected stabilization time of a population protocol.
For example, consider an execution
$\Xi_P(C_0,\Gamma)=C_0,C_1,\dots$ where agents in $V'$
have different values in variable $\var$ in configuration $C_0$
and the larger value is propagated
from agent to agent whenever two agents in $V'$
have an interaction.
Clearly, all agents in $V'$
have the maximum value of $\var$
when all agents in $V'$ are infected in one-way epidemic
in $V'$ and under $\Gamma$
starting from the agent with the maximum value $\var$
in configuration $C_0$.

Angluin et al.~\cite{original} prove that
one-way epidemic in the whole population $V$
from any agent $r \in V$ finishes
(\ie all agents are infected) within $\Theta(n \log n)$ interactions
with high probability.
Furthermore, Sudo et al.~\cite{kanjiko}
give a concrete lower bound on the probability that the epidemic
in the whole population
finishes within a given number of interactions. 
We generalize this lower bound for an epidemic in any set of agents (sub-population) $V' \subseteq V$ as follows while
the proof is almost the same as the one in \cite{kanjiko}.

\begin{lemma}
\label{lemma:epidemic}
Let $V' \subseteq V$, $r \in V'$, $n'=|V'|$,
and $t \in \mathtt{N}$.
We have 
$\Pr(I_{V',r,\Gamma}(2\lceil n/n' \rceil t) \neq V') \le n e^{-t/n}$.
\end{lemma}

\begin{proof}
For each $k\ (2\le k \le n')$,
we define $T(k)$ as integer $t$ such that
$|I_{V',r,\Gamma}(t-1)| = k-1$ and 
$|I_{V',r,\Gamma}(t)| = k$,
and define $T(1) = 0$.
Intuitively, $T(k)$ is the number of interactions required to
infect $k$ agents in $V'$.
Let $X_\mathrm{pre} = T(\lceil \frac{n'+1}{2} \rceil)$
and $X_\mathrm{post} = T(n') - T(n' - \lceil \frac{n'+1}{2} \rceil + 1)$.

Let $k$ be any integer such that $1 \le k \le n'$.
When $k$ agents are infected,
an agent is newly infected with probability $k(n'-k)/_n C_2$ at every step.
When $n'-k$ agents are infected,
an agent is newly infected also with probability $k(n'-k)/_n C_2$ at every step.
Therefore, $T(k+1)-T(k)$ and $T(n'-k+1)-T(n'-k)$ have the same probability distribution.
Thus, $X_\mathrm{pre} = T(\lceil \frac{n'+1}{2} \rceil) = \sum_{j=1}^{\lceil (n'+1)/2 \rceil-1}T(j+1)-T(j)$
and $X_\mathrm{post} = T(n')-T(n'-\lceil \frac{n'+1}{2}\rceil+1)=\sum_{j=1}^{\lceil (n'+1)/2 \rceil-1} T(n'-j+1)-T(n'-j)$
have the same probability distribution. 
Moreover, $X_\mathrm{pre} + X_\mathrm{post} \ge T(n')$ holds because
$\lceil \frac{n'+1}{2} \rceil \ge n' - \lceil \frac{n'+1}{2} \rceil + 1$.

In what follows, we bound the probability that $X_\mathrm{post} > \lceil n/n'\rceil t$. 
We denote $T(n'-\lceil \frac{n'+1}{2} \rceil + 1)$ by $T_\mathrm{half}$.
For any agent $v \in V$,
let $T_v$ be the minimum non-negative integer such that $v \in I_{V',r,\Gamma}(T_v)$,
\ie agent $v$ becomes infected at the $T_v$-th step.
We define $X_v = \max(T_{C_0,\Gamma}(v)-T_\mathrm{half},0)$.
Consider the case  $v \notin I_{V',r,\Gamma}(T_\mathrm{half})$.
At any step $t \ge T_\mathrm{half}$,
at least $n'-\lceil \frac{n'+1}{2}\rceil + 1\ (\ge \frac{n'}{2})$
agents are infected.
Therefore, each interaction $\Gamma_t$ such that $(t \ge T_\mathrm{half})$ infects $v$
with the probability at least $\frac{1}{{}_n C_2}\cdot\frac{n'}{2} > \frac{n'}{n^2}$,
hence we have
$\Pr(X_v >\lceil n/n' \rceil t) \le \left(1 -\frac{n'}{n^2}\right)^{ nt/n'} \le e^{-t/n}$.
Since the number of non-infected agents at step $T_\mathrm{half}$
is at most $n'/2$,
$
\Pr(X_\mathrm{post} >\lceil n/n' \rceil t)
\le  \Pr(\bigvee_{v\in V} (X_v >\lceil n/n' \rceil t)) 
\le \frac{n'}{2}\cdot e^{-t/n}$ holds.

By the equivalence of the distribution of
$X_\mathrm{pre}$ and $X_\mathrm{post}$,
we have
\begin{align*}
\Pr\left(I_{V',r,\Gamma}\left(2\lceil n/n'\rceil t \right)
\neq V' \right)
\le \Pr\left(X_\mathrm{pre} > \lceil n/n'\rceil t \right)
+ \Pr\left(X_\mathrm{post} > \lceil n/n'\rceil t \right)
\le n e^{-t/n}.
\end{align*} 
\end{proof}

\section{Logarithmic Leader Election}
\label{sec:pll}
\subsection{Key Ideas}
\label{sec:key}
In this subsection, we give key ideas
of the proposed protocol $\plog$.
Each agent $v$ keeps output variable $v.\leader \in \{\fl,\tr\}$.
An agent outputs $L$ when the value of $\leader$ is $\tr$
and it outputs $F$ when it is $\fl$.
An execution of $\plog$ can be regarded as a competition
by agents.
At the beginning of the execution,
every agent has $\leader=\tr$, that is,
all agents are leaders.
Throughout the execution,
every leader tries to remain a leader
and tries to make all other leaders followers
so that it becomes the unique leader in the population.
The competition consists of
three modules $\quick()$, $\tourn()$,
and $\backup()$,
which are executed in this order.
These three modules guarantees the following properties:
\begin{description}
 \item[$\quick()$:]
An execution of this module takes $O(\log n)$ parallel time
in expectation.
For any $i \ge 2$,
exactly $i$ leaders survive an execution of $\quick()$
with probability at most $2^{1-i}$.
The execution never eliminates all leaders,
i.e., at least one leader always survives.
 \item[$\tourn()$:]
An execution of this module takes $O(\log n)$ parallel time
in expectation.
By an execution of $\tourn()$, which starts with $i\ge 2$ leaders,
the unique leader is elected with probability at least
$1-O(i/\log n)$.
This lower bound of probability is independent of
an execution of the previous module $\quick()$.
The execution never eliminates all leaders,
i.e., at least one leader always survives.
\item[$\backup()$:]
An execution of this component elects a unique leader
within $O(\log^2 n)$ parallel time in expectation.
\end{description}

From above, it holds that, after executions of
$\quick()$ and $\tourn()$ finish, 
the number of leaders is exactly one
with probability at least
$1-\sum_{i=2}^{n}O\left(\frac{i}{2^{i-1}\log n}\right)=1-O(1/\log n)$.
Therefore, combined with $\backup()$,
protocol $\plog$ elects a unique leader
within $(1-O(1/\log n))\cdot \log n + O(1/\log n)\cdot O(\log ^2 n)=O(\log n)$ parallel time
in expectation.

In the remainder of this subsection,
we briefly give key ideas to design the three modules
satisfying the above guarantees.
We will present a way to implement the following ideas
with $O(\log n)$ states per agent in the next subsection (Section \ref{sec:detailed}). In this subsection, keep in mind only that
these ideas are easily implemented with poly-logarithmic number of states per agent, that is, with a constant number of variables
with $O(\log \log n)$ bits. For the following description of the key ideas, we assume a kind of global synchronization,
for example, we assume that each agent begins an execution
of $\tourn()$ after \emph{all} agents finish necessary operations
of $\quick()$. We also present a way to implement such
a synchronization in Section \ref{sec:detailed}.

\subsubsection{Key Idea for $\quick()$}
\label{sec:keyquick}
The goal of this module
is to reduce the number of leaders
such that, for any $i \ge 2$,
the resulting number of leaders is exactly $i$
with probability at most $2^{1-i}$
while guaranteeing that
not all leaders are eliminated.
This module is based on
almost the same idea as \emph{the lottery protocol}
in \cite{AAE+17}.
The protocol $\plog$ achieves much faster stabilization time than the lottery protocol
thanks to tighter analysis on the number of surviving
leaders, which we will see below,
and the combination with the other two modules. 

First, consider the following game:
\begin{itemize}
 \item (i) Each agent in $V$ executes
a sequence of independent fair coin flips,
each of which results in head with probability
$1/2$ and tail with probability $1/2$,
until it observes tail for the first time,
 \item (ii) Let $s_v$ be the number of heads that $v$ observes
in the above coin flips and let $\smax=\max_{v \in V}s_v$
 \item (iii) The agents $v$ with $s_v = \smax$ are winners
and the other agents are losers.
\end{itemize}
Let $i \ge 2$ and $j \ge 0$.
Consider the situation that exactly $i$ agents
observe that their first $j$ coin flips
result in head and define $p_{i,j}$ as the probability that
all the $i$ agents wins the game in the end
starting from this situation.
Starting from this situation,
if all the $i$ agents observe tail in their $j+1$-st coin flips
then exactly $i$ agents win the game
with probability $1$;
if all the $i$ agents observe head in their $j+1$-st coin flips
then exactly $i$ agents win with probability $p_{i,j+1}$;
Otherwise, the number of winners of the game is
less than $i$ with probability $1$.
Therefore, we have $p_{i,j} = 2^{-i}+2^{-i}\cdot p_{i,j+1}$.
Since we have $p_{i,j}=p_{i,j+1}$ thanks to memoryless property
of this game, solving this equality gives $p_{i,j}=1/(2^i-1) \le 2^{1-i}$.
Let $k_i$ be the minimum integer $j$ such that
exactly $i$ agents
observes that all of their first $j$ coin flips
result in head.
We define $k_n=0$ for simplicity.
Then, for any $i \ge 0$, we have
\begin{align*}
\Pr(|\{v \in V \mid s_v = \smax\}|=i)
= \sum_{j=0}^{\infty}
\Pr(k_i = j)
\cdot p_{i,j}
\le 2^{1-i} \sum_{j=0}^{\infty} \Pr(k_i = j)
\le 2^{1-i}.
\end{align*}

Module $\quick()$ simulates this game
in the population protocol model.
Every time an agent $v$ has an interaction,
we regard the interaction as the coin flip by $v$.
If $v$ is an initiator at the interaction, 
we regard the result of the coin flip as head;
Otherwise we regard it as tail.
The correctness of this simulation for coin flips
comes from the definition of the uniformly random
scheduler: at each step, 
an interaction where $v$ is an initiator
happens with probability $1/n$
and an interaction where $v$ is a responder
also happens with probability $1/n$.
Strictly speaking, this simple simulation of coin flips
does not guarantee independence of coin flips by
$u$ and $v$
for any distinct $u,v \in V$.
However, the actual $\plog$ defined
in Section \ref{sec:detailed}
completely simulates independent coin flips of leaders
and we will explain it in Section \ref{sec:detailed}.
Each agent $v$ computes and stores $s_v$ on variable $v.\levelq$
by counting 
the number of interactions that it participates in
as an initiator
until it interacts as a responder for the first time.
After every agent $v$ computes $s_v$ on $v.\levelq$,
the maximum value of $\levelq$,
i.e., $\smax$, is propagated from agent to agent
via \emph{one-way epidemic} \cite{fast},
that is,
\begin{itemize}
 \item each agent memorizes the largest value of $\levelq$ it has observed, and
 \item the larger value is propagated to
the agent with smaller value
at every interaction.
\end{itemize}
It is proven in \cite{fast}
that all agents obtain the largest value within $O(\log n)$ parallel time with high probability by this simple propagation.
If agent $v$ knows $s_v < \smax$, 
$v$ changes $v.\leader$ from $\tr$ to $\fl$, that is,
$v$ becomes a follower.
Thus, when one-way epidemic of $\smax$ finishes,
only the agents $v$ satisfying $s_v=\smax$
are leaders.
From the above discussion, for any $i \ge 2$,
the number of such surviving leaders is exactly $i$
with probability at most $2^{1-i}$.
On the other hand, there are at least one agent $v$
with $s_v = \smax$, thus this module never
eliminates all leaders.
A logarithmic number of states is sufficient
for $\levelq$ 
because each agent $v$
gets more than $c \lg n$ consecutive heads
with probability at most $n^{-c}$ for any $c \ge 1$.

\subsubsection{Key Idea for $\tourn()$}
\label{sec:keytourn}
Starting from a configuration
where the number of leaders is $i$,
the goal of $\tourn()$ is to 
reduce the number of leaders
from $i$ to one with probability $1-O(i/ \log n)$
while guaranteeing that
not all leaders are eliminated.
The idea of this component is simple.
As with the $\quick()$,
we use coin flips in $\tourn()$.
Every leader $v$ maintains variable $v.\rand$.
Initially, $v.\rand = 0$.
Every time it has an interaction,
it updates $v.\rand$ by
$v.\rand \gets 2v.\rand + j$ where
$j$ indicates whether $v$ is a responder in the interaction
or not,
\ie $j=0$ if $v$ is a initiator
and $j=1$ if $v$ is a responder.
This operation stops when $v$ encounters
$\clog{m} = O(\log \log n)$ interactions.
Thus, when all the $i$ leaders encounter
at least $\clog{m}$ interactions,
for every leader $v$,
$v.\rand$ is a random variable
uniformly chosen
from $\{0,1,\dots,2^{\clog{m}}-1\}$.
Although $u.\rand$ and $v.\rand$
are not independent of each other for any distinct leader
$u$ and $v$, 
we will present a way to remove any dependence
between $u.\rand$ and $v.\rand$ in Section \ref{sec:detailed}.
As with $\quick()$, the maximum value $\rand$ is propagated
to the whole population via one-way epidemic
within $O(\log n)$ parallel time with high probability
and only leaders with the maximum value remains leaders
in the end of $\tourn()$.

Let $v_1,v_2,\dots,v_i$ be the $i$ leaders
that survive $\quick()$.
Let $r_1,r_2,\dots,r_i$ be the resulting value
of $v_i.\rand$
and define $\rmax(j)=\max(r_1,r_2,\dots,r_j)$ for any
$j = 1,2,\dots,i$.
Clearly, the number of leaders at the end of $\tourn()$
is exactly one if $r_{j+1} \neq \rmax(j)$ holds
for all $j=1,2,\dots,i-1$.
By the union bound and independence between $r_1,r_2,\dots,r_i$,
this holds with probability at least
$1-\sum_{j=1}^{i-1} 2^{-\clog{m}} \ge 1-i/m \ge 1-i/(\lg n)$.
On the other hand, an execution of $\tourn()$ never eliminates all leaders since there always at least one leader $v_j$ that satisfies $r_j = \rmax(i)$.

\subsubsection{Key Idea for $\backup()$}
\label{sec:keybackup}
The goal of $\backup()$
is to elect a unique leader within
$O(\log^2 n)$ parallel time
in expectation.
We must guarantee this expected time
regardless of the number of the agents
that survive both $\quick()$ and $\tourn()$
and remain leaders at the beginning
of an execution of $\backup()$.
We can only assume that at least one leader
exists at the beginning of the execution.
We use coin flips also for $\backup()$.
Every leader $v$ maintains $v.\levelb$.
Initially, $v.\levelb = 0$.
Every leader $v$ repeats the following procedure
until $v.\levelb$ reaches $5m$
or $v$ becomes a follower.
\begin{itemize}
 \item Make a coin flip. If the result is head
(i.e., $v$ participates in an interaction as an initiator),
$v$ increments $v.\levelb$ by one. If the result is tail,
$v$ does nothing.
 \item Wait for sufficiently long but logarithmic parallel time
so that the maximum $\levelb$ propagates to the whole population
via one-way epidemic. If it observes larger value in the epidemic,
it becomes a follower, that is, it executes $v.\leader \gets \fl$.
Furthermore, if $v$ interacts with another leader with the same level during this period and $v$ is a responder in the interaction, 
$v$ becomes a follower.
\end{itemize}
Let $j$ be an arbitrary integer such that $1 \le i \le 5m$.
Consider the first time that $\levelb$ of some leader, say $v$, reaches $j$.
Let $V' \subseteq V$ be the set of leaders at that time.
By the definition of the above procedure,
every $u \in V'$ other than $v$ satisfies $u.\levelb < j$,
and $u$ makes a coin flip at most once
with high probability
until the maximum value $j$ is propagated from $v$ to $u$.
If the result of the one coin flip is tail,
$u$ becomes a follower.
Therefore, with probability at least $1/2-O(n^{-1})>1/3$,
no less than half of leaders in $V'\setminus{v}$
becomes followers,
that is, the number of leaders
decreases to at most $1+\lfloor|V'|/2 \rfloor$.
Chernoff bound guarantees that
the number of leaders becomes one
with high probability
until $v.\levelb$ for every leader $v$ reaches $5m$.
Even if multiple leaders survive at that time,
we have simple election mechanism to elect a unique leader;
when two leaders with the same level interacts with each other,
one of them becomes a follower.
This simple election mechanism elects a unique leader
within $O(n)$ parallel time in expectation.
Therefore, the total expected parallel time to elect a unique leader is $O(m \log n) + O(n^{-1})\cdot O(n)= O(\log^2 n)$.

\subsection{Detailed Description}
\label{sec:detailed}
In this subsection, we present detailed description of the proposed protocol $\plog$.
The key ideas
presented in the previous subsection
achieve $O(\log n)$ stabilization time
if it is implemented correctly.
However, they need some kind of global
synchronization.
Furthermore, a naive implementation of the key ideas requires a poly-logarithmic number of states (\ie $O(\log^c n)$ states for $c > 1$) per agent while our goal is to achieve $O(\log n)$ states per agent.
In this subsection, we will give how we achieve synchronization and implement the ideas shown in Section \ref{sec:key} with only $O(\log n)$ states per agent.

\begin{table}
\caption{Variables of $\plog$}
\label{tbl:variables}
\center
\begin{tabular}{|l|l|l|}
 \hline
 Groups & Variables & Initial values\\ \hline 
 \multirow{6}{*}{All agents}
 &  $\leader \in \{\fl,\tr\}$ & $\tr$ \\ \cline{2-3}
 &  $\tick \in \{\fl,\tr\}$ & $\fl$ \\ \cline{2-3}
 &  $\status \in \{\initial,\cand,\timer\}$ & $\initial$ \\ \cline{2-3}
 &  $\epoch \in \{1,2,3,4\}$ & $1$ \\ \cline{2-3}
 &  $\init \in \{1,2,3,4\}$ & $1$ \\ \cline{2-3}
 &  $\clr \in \{0,1,2\}$ & $0$ \\ \hline
 $\vtimer$ & $\cnt \in \{0,1,\dots,\cmax-1\}$ & Undefined\\ \hline
 \multirow{2}{*}{$\vcand \cap V_1$}
 & $\levelq \in \{0,1,\dots,\lmax\}$ & Undefined \\ \cline{2-3}
 & $\done \in \{\fl,\tr\}$ & Undefined \\ \hline
 \multirow{2}{*}{$\vcand \cap (V_2 \cup V_3)$}
 & $\rand \in \{0,1,\dots,2^{\rsize}-1\}$ & Undefined \\ \cline{2-3}
 & $\ind \in \{0,1,\dots,\rsize-1\}$ & Undefined\\ \hline
 $\vcand \cap V_4$& $\levelb \in \{0,1,\dots,\lmax\}$ & Undefined\\ \hline
\end{tabular}
\end{table}

All variables of $\plog$ are listed in Table \ref{tbl:variables}.
All agents manage six variables
$\leader$, $\tick$, $\status$, $\epoch$, $\init$, and $\clr$.
To implement the key ideas above with $O(\log n)$ states,
we divide the population into multiple sub-populations
or \emph{groups}, as in \cite{GSU18},
where agents in different groups manage different variables
in addition to the above six variables.
In the remainder of this paper, 
we refer the above six variables by \emph{common variables}
and other variables by \emph{additional variables}.
The population is divided to six groups 
based on two common variables $\status \in \{\initial,\cand,\timer\}$ and $\epoch \in \{1,2,3,4\}$,
that is,
$\vinitial$, $\vtimer$, $\vcand \cap V_1$,
$\vcand \cap (V_2 \cup V_3)$, $\vcand \cap V_4$
where we denote $V_Z= \{v \in V \mid v.\status = Z\}$
for $Z \in \{\initial,\cand,\timer\}$
and $V_i=\{v \in V \mid v.\epoch = i\}$
for $i \in \{1,2,3,4\}$.
We have no additional variables for agents
in group $\vinitial$,
one additional variable $\cnt \in \{0,1,\dots,\cmax-1\}$
for agents in $\vtimer$
where $\cmax = 41 m$,
two additional variables $\levelq \in \{0,1,\dots,\lmax\}$
and $\done \in \{\fl,\tr\}$ for agents in $\vcand \cap V_1$
where $\lmax = 5m$,
two additional variables $\rand \in \{0,1,\dots,2^\rsize-1\}$
and $\ind \in \{0,1,\dots,\rsize - 1\}$
for agents in $\vcand \cap (V_2 \cup V_3)$
where $\rsize = \lceil \frac{2}{3} \lg m \rceil$,
and
one additional variable $\levelb \in \{0,1,\dots,\lmax\}$
for agents in $\vcand \cap V_4$.
Agents in any group have only $O(\log n)$ states.
This is because every common variable has constant size domain,
every group other than $\vcand \cap (V_2 \cup V_3)$
has at most one non-constant additional variable
and any of such variables can take $O(\log n)$ values, 
and an agent in $\vcand \cap (V_2 \cup V_3)$
has two additional variables $\rand$ and $\ind$
and the combination of the two variables can take
$2^{\rsize} \cdot \rsize = O(m^{2/3} \log m) \subset O(\log n)$ values.
Therefore, the number of states per agent used by $\plog$ is
$O(\log n)$.

\begin{lemma}
The number of states per agent used by $\plog$ is $O(\log n)$.
\end{lemma}

\begin{algorithm}[t]
\caption{$\plog$}
\label{al:pll}
\textbf{Notations}:
\begin{algorithmic}
\STATE $\lmax = 5 m$, $\cmax = 41 m$, $\rsize = \lceil \frac{2}{3}\lg m \rceil$
\STATE $V_Z=\{v \in V \mid v.\status = Z\}$
for $Z \in \{\initial,\cand,\timer\}$
\STATE $V_i=\{v \in V \mid v.\epoch = i\}$
for $i \in \{1,\dots,4\}$
\end{algorithmic}
\textbf{Output function} $\outputs$:
\begin{algorithmic}
\STATE if $v.\leader = \tr$ holds, then the output of agent $v$ is $L$, otherwise $F$.
\end{algorithmic}
\textbf{Interaction} between initiator $a_0$ and responder $a_1$:
\hspace{-1cm} \verb| |
\begin{algorithmic}[1]  
 \IF{$a_0,a_1 \in \vinitial$} 
 \STATE $(a_0.\status,a_0.\levelq,a_0.\done,a_0.\leader) \gets (\cand,0,\fl,\tr)$
 \STATE $(a_1.\status,a_1.\cnt,a_1.\leader)\gets(\timer,0,\fl)$
 \ELSIF{$\exists i \in \{0,1\}:a_i \in \vinitial \wedge a_{1-i} \notin \vinitial$}
 \STATE $(a_i.\status,a_i.\levelq,a_i.\done,a_i.\leader) \gets (\cand,0,\tr,\fl)$
 \ENDIF
 \vspace{0.3cm}
 \STATE $a_0.\tick \gets a_1.\tick \gets \fl$
 \STATE $\countup()$
 \STATE \foralldo{$i\in \{0,1\}$ such that $a_i.\tick$}{$a_i.\epoch = \max(a_i.\epoch+1,4)$}
 \STATE $a_0.\epoch \gets a_1.\epoch \gets \max(a_0.\epoch,a_1.\epoch)$
\vspace{0.3cm}
 \FORALL[Initialize variables for each group]{$i \in \{0,1\}$ such that $a_i.\epoch > a_i.\init$}
 \STATE \ifthen{$a_i \in \vcand \cap (V_2 \cup V_3)$}{$(a_i.\rand,a_i.\ind)\gets(0,0)$}
 \STATE \ifthen{$a_i \in \vcand \cap V_4$}{$a_i.\levelb \gets 0$}
 \STATE $a_i.\init \gets a_i.\epoch$
 \ENDFOR 
\vspace{0.3cm}
 \IF{$a_0,a_1 \in V_1$}
 \STATE Execute $\quick()$
 \ELSIF{$a_0,a_1 \in V_2 \vee a_0,a_1 \in V_3$}
 \STATE Execute $\tourn()$
 \ELSIF{$a_0,a_1 \in V_4$}
 \STATE Execute $\backup()$
 \ENDIF
\end{algorithmic}
\end{algorithm}

Independently of the six groups defined above,
we define another groups $\vl$ and $\vf$
based on a common variable $\leader$;
$\vl$ (resp., $\vf$) is the set of agents $v \in V$
such that $v.\leader=\tr$ (resp., $v.\leader = \fl$).
We introduce these two groups only for simplicity of notation.

The pseudo code of $\plog$ is given in Algorithm \ref{al:pll}
and its modules $\countup()$, $\quick()$, $\tourn()$, and $\backup()$ are presented in Algorithm \ref{al:countup}, \ref{al:quick}, \ref{al:tourn}, and \ref{al:backup}, respectively.
The main function of $\plog$ (Algorithm \ref{al:pll})
consists of four parts.
The first part (Lines 1-6) assigns status $\cand$ or $\timer$
to each agent.
The second part (Lines 7-10) manages variable $\epoch$
using module $\countup()$.
Initially, $v.\epoch = 1$ holds, that is, $v \in V_1$
holds for all $v \in V$.
In an execution of $\plog$,
$v.\epoch$ never decreases and increases by one
every sufficiently large logarithmic parallel time in expectation
until it reaches $4$ as we will explain later.
In the third part (Lines 11-15),
we initialize additional variables
when an agent increases its epoch. 
Each agent $v$ has a common variable $\init$,
which is set to $1$ initially.
Whenever $v.\epoch$ increases,
$v.\epoch > v.\init$ must hold,
then $v$ initialize additional variables
according to $v$'s group and executes $v.\init \gets v.\epoch$.
For example, when the $\epoch$ of agent $v \in \vcand$
changes from $3$ to $4$
\ie $v$ moves from group $\vcand \cap V_3$ to $\vcand \cap V_4$,
it initializes an additional variable $a_i.\levelb$
to $0$ (Line 13).
Additional variables for groups $\vtimer$ and $\vcand \cap V_1$
are initialized not in this part but in the first part as we will explain
in Section \ref{sec:status}.
In the fourth part (Lines 16-22),
agents execute modules based on the values of
their $\epoch$.
Specifically,
agents execute $\quick()$, $\tourn()$, and $\backup()$
while they are in $V_1$, $V_2\cup V_3$, and $V_4$ respectively.


In the remainder of this subsection,
we explain how $\plog$ assigns status to agents,
$\plog$ synchronizes the population by $\countup()$, 
and the implementation of the three modules
$\quick()$, $\tourn()$, and $\backup()$.

\subsubsection{Assignment of Status}
\label{sec:status}
At the beginning of an execution,
all agents are in $\vinitial$, that is,
the statuses of all agents are the ``initial'' status $\initial$.
Every agent is given status $\cand$ or $\timer$
at its first interaction where $\cand$ means ``leader candidate''
and $\timer$ means ``timer agent''.
As we will explain later,
the unique leader is elected from $\vcand$
and agents in $\vtimer$ are mainly used to synchronize the population
with their count-up timers.

Agents determine their status, $\cand$ or $\timer$,
by the following simple way.
When two agents in $\vinitial$ meet,
the initiator and the responder are given status $\cand$
and $\timer$, respectively (Line 2-3).
The initiator initializes its additional variable
$\levelq$ and $\done$ to $0$ and $\fl$ respectively
and remains a leader (Line 2)
while the responder
initializes its additional variable $\cnt$ to 0
and becomes a follower by $\leader \gets \fl$ (Line 3).
When an agent in $\vinitial$ meets an agent in $\vcand$ or $\vtimer$,
it gets status $\cand$ but it becomes a follower.
It also initialize its additional variable 
$\levelq$ and $\done$ to $0$ and $\tr$ respectively (Line 5).
For agent $v$, assigning $\tr$ to $v.\done$ means that
$v$ never joins a game with coin flips in $\quick()$.

No agent changes its status
once it gets status $\cand$ or $\timer$,
and no follower becomes a leader
in an execution of $\plog$.
Therefore, we have the following lemma.

\begin{lemma}
\label{lemma:status} 
In an execution of $\plog$,
$|\vcand| \ge n/2$, $|\vf| \ge n/2$, and $|\vtimer|\ge 1$ always hold
after every agent 
gets status $\cand$ or $\timer$.
\end{lemma}
\begin{proof}
Consider any configuration in $\call(\plog)$
where every agent has status $\cand$ or $\timer$.
Let $x$ (resp., $y$ and $z$) be the the number of agents which get status $\cand$ (resp., $\timer$ and $\cand$) by Line 2
(resp., Line 3 and Line 5).
We have $x=y \le n/2$ by the definition of $\plog$,
which gives $|\vcand| = x + z = n-y \ge n/2$.
Moreover, $|\vl| \le x \le n/2$ holds
because the number of leaders is monotonically non-increasing
in an execution of $\plog$.
The first interaction of the execution assigns
one agent with status $\vtimer$, hence
$|\vtimer|\ge 1$ holds.
%
%
%
\end{proof}


\begin{algorithm}[t]
\caption{$\countup()$}
\label{al:countup}
\textbf{Interaction} between initiator $a_0$ and responder $a_1$:
\hspace{-1cm} \verb| |
\begin{algorithmic}[1]
 \setcounter{ALC@line}{22}  
\FORALL{$i \in \{0,1\}$ such that $a_i \in \vtimer$}
\STATE $a_i.\cnt \gets a_i.\cnt + 1 \pmod{\cmax}$
\IF{$a_i.\cnt=0$}
\STATE $a_i.\clr \gets a_i.\clr + 1 \pmod{3}$
\STATE $a_i.\tick \gets \tr$
\ENDIF
\ENDFOR
\IF{$\exists i \in \{0,1\}: a_{1-i}.\clr = a_i.\clr+1 \pmod{3}$}
\STATE $a_i.\clr \gets a_{1-i}.\clr$
\STATE $a_i.\tick \gets \tr$
\STATE \ifthen{$a_i \in \vtimer$}{$a_i.\cnt \gets 0$}
\ENDIF
\end{algorithmic}
\end{algorithm}

\subsubsection{Synchronization and Epochs}
\label{sec:sync}
When a unique leader exists in the population, 
we can synchronize the population by Phase clocks
with constant space per agent \cite{fast}.
Recently, in \cite{GS18} and \cite{GSU18},
it is proven that
even when we cannot assume the existence of the unique leader,
Phase clocks can be used for synchronization
if we are allowed to use $O(\log \log n)$ states per agent.
Since we use $O(\log n)$ states for another modules, 
we achieve synchronization in simpler way with $O(\log n)$ states per agent.


For synchronization,
we use common variables
$\clr \in \{0,1,2\}$ in all agents
and an additional variable
$\cnt \in \{0,1,\dots,\cmax-1\}$ for agents in group $\vtimer$.
Initially, all agents have the same color, namely, $0$.
The color of an agent is incremented by modulo $3$
when the agent changes its color.
We say that the agent \emph{gets a new color}
when this event happens.
Roughly speaking, our goal is to guarantee that
\begin{itemize}
 \item (i) whenever one agent gets a new color (\eg changes its $\clr$ from $0$ to $1$), the new color spreads to the whole population
within $O(\log n)$ parallel time with high probability,
 \item (ii) thereafter, all agents keeps the same color
for sufficiently long but $\Theta(\log n)$ parallel time with high probability.
\end{itemize}
Specifically, ``sufficiently long but $\Theta(\log n)$
time'' in (ii)
means 
sufficiently long period 
such that
any $O(\log n)$ parallel time operations
in $\quick()$, $\tourn()$, and $\backup()$,
such as one-way epidemic of some value,
finishes with high probability during the period.

At every interaction, module $\countup()$ is invoked (Line 8)
and 
variables $\clr$ and $\cnt$ can be changed only in this module.
In $\countup()$, every agent in $\vtimer$ increments
its $\cnt$ by one modulo $\cmax$ (Line 24).
For every $v \in \vtimer$,
if this incrementation changes $v.\cnt$ from $\cmax-1$ to $0$,
$v$ gets a new color by incrementing $v.\clr$ by one modulo $3$
(Line 26).
Once one agent gets a new color,
the new color spreads to the whole population
via one-way epidemic in the whole population. 
Specifically, if agents $u$ and $v$ satisfying $u.\clr=v.\clr+1 \pmod{3}$ meets, $v$ execute $v.\clr \gets u.\clr$ and resets
its $\cnt$ to $0$ (Line 31-33).


Every time an agent $v$ gets a new color,
it raise a tick flag, i.e., assigns $v.\tick \gets \tr$
 (Lines 27 and 32).
This common variable $v.\tick$
is used only for simplicity of the pseudo code
and it does not affect the transition
at $v$'s next interaction ($v.\tick$ is reset to $\fl$ in Line 7), unlike any other variable.
When $v.\tick$ is raised,
$v.\epoch$ increases by one unless it has already reached $4$
(Line 9).
After two agents $u$ and $v$
execute Lines 7-9 at an interaction,
$u.\epoch = v.\epoch$ usually holds.
However, this equation does not hold
when synchronization fails.
For this case, we substitute $\max(u.\epoch,v.\epoch)$
into $u.\epoch$ and $v.\epoch$ in Line 10.

As mentioned above, every agent gets a new color
in every sufficiently large $\Theta(\log n)$ parallel time
with high probability.
This means that, for every $v \in V$,
$v.\tick$ is raised and $v.\epoch$ increases
by one with high probability in
every sufficiently large $\Theta(\log n)$ parallel time
until $v.\epoch$ reaches $4$.
If this synchronization fails,
\eg some agent gets a color $1$ without keeping color $0$
for $\Theta(\log n)$ parallel time,
the modules $\quick()$ and $\tourn()$
may not work correctly.
However,
starting from any configuration
after a synchronization fails arbitrarily,
module $\countup()$ and Lines 7-10 guarantees that
all agents proceeds to the forth epoch
within $O(\log n)$ parallel time in expectation,
and thereafter
$\backup()$ guarantees that
exactly one leader is elected within $O(n)$ parallel time
in expectation.
Hence, $\plog$ guarantees that
a unique leader is elected with probability $1$.
The above $O(n)$ parallel time
never prevent us from achieving stabilization
time of $O(\log n)$ parallel time in expectation because
synchronization fails with probability
at most $O(\log n/n)$ 
as we will see later.


\begin{definition}
For any $i = 0,1,2$, we define $\ccolor(i)$
as the set of all configurations in $\call(\plog)$
where every agent has color $i$.
\end{definition}
\begin{definition}
For any $i = 0,1,2$, we define $\cstart(i)$
as the set of all configurations in $\call(\plog)$
each of which satisfies all of the following conditions:
\begin{itemize}
 \item some agent has color $i$,
 \item $v.\cnt=0$ holds for all $v \in \vtimer$ such that
       $v.\clr = i$, and
 \item no agent has color $i+1 \pmod{3}$.
\end{itemize}
\end{definition}
\begin{lemma}
\label{lemma:timer}
In an execution of $\Xi(\cinit{\plog},\Gamma)$,
each agent in $\vtimer$ always
gets a new color within $O(\log n)$ parallel time
with high probability.
\end{lemma}
\begin{proof}
Simple Chernoff bound gives the lemma
because
any agent has an interaction with
probability $2/n$ at each step
and each agent in $\vtimer$ gets a new color before
it has $\cmax$ interactions.
\end{proof}
The goal of our synchronization, (i) and (ii),
are formalized as follows.
\begin{lemma}
\label{lemma:sync}
Let $i \in \{0,1,2\}$, $C_0 \in \cstart(i)$,
and $\Xi_{\plog}(C_0,\Gamma)=C_0,C_1,\dots$.
Then, all of the following propositions hold.
\begin{itemize}
 \item $P_1$: No agent gets color $i+1 \pmod{3}$
by the first $\unit$ steps
in $\Xi_{\plog}(C,\Gamma)$ 
with high probability.
 \item $P_2$: Execution $\Xi_{\plog}(C_0,\Gamma)$
reaches a configuration in $\ccolor(i)$
by the first $\lfloor 4n \ln n \rfloor$ steps
with high probability.
 \item $P_3$: Execution $\Xi_{\plog}(C_0,\Gamma)$
reaches a configuration in $\cstart(i+1 \pmod{3})$
within $O(\log n)$ parallel time
with high probability.
\end{itemize}
\end{lemma}
\begin{proof}
Propositions $P_2$ and $P_3$ immediately follows from
Lemma \ref{lemma:epidemic} with $n'=n$
and Lemma \ref{lemma:timer}, respectively.
In the following,
we prove proposition $P_1$.
Starting from
a configuration $C_0 \in \cstart(i)$,
no agent gets color $i+1 \pmod{3}$
until some agent in $\vtimer$ participates in
no less than $\cmax$ interactions.
For any agent $v$,
$v$ participates in an interaction
with probability $2/n$ at every step.
Therefore, letting $X$ be
a binomial random variable such that
$X \sim B(\unit,2/n)$,
$v$ participates in no less than
$\cmax$ interactions with probability
$\Pr(X \ge \cmax)$, which is bounded as follows.
\begin{align*}
\Pr(X \ge \cmax)
&= \Pr\left(X \ge \frac{\cmax}{42 \ln n}\ex[X]\right)\\
&\le \Pr\left(X \ge \frac{58}{42}\ex[X]\right)\\
&\le
\exp\left(-\frac{(58-42)^2}{42^2\cdot 3} \ex[X]\right)\\
&\le
\exp(-2 \ln n+0.05)\\
&=O\left(n^{-2}\right).
\end{align*}
where we use $\cmax \ge 41 \lg n \ge 58 \ln n$
for the second inequality
and Chernoff Bound in the form of
\eqref{eq:upperdouble} in Lemma \ref{lemma:chernoff}
for the third inequality.
Thus, the union bound gives that
no agent gets color $i+1 \pmod{3}$ by the first
$\unit$ interactions
in $\Xi_{\plog}(C,\Gamma)$
with probability $1-O(n^{-1})$.
\end{proof}
\begin{algorithm}[t]
\caption{$\quick()$}
\label{al:quick}
\textbf{Interaction} between initiator $a_0$ and responder $a_1$:
\hspace{-1cm} \verb| |
\begin{algorithmic}[1]  
 \setcounter{ALC@line}{34}
 \IF{$\exists i \in \{0,1\}$ such that $a_i \in \vl \wedge a_{1-i} \in \vf \wedge \lnot a_i.\done$}
 \STATE
 \ifthen{$i=0$}{$a_0.\levelq \gets \max(a_0.\levelq+1,\lmax)$}
 \STATE \ifthen{$i=1$}{$a_1.\done \gets \tr$}
 \ENDIF
 \IF{$a_0,a_1 \in \vcand \wedge a_0.\done \wedge  a_1.\done  \wedge \exists i \in \{0,1\}: a_i.\levelq < a_{1-i}.\levelq$}
 \STATE $a_i.\leader \gets \fl$
 \STATE $a_i.\levelq \gets a_{1-i}.\levelq$
 \ENDIF
\end{algorithmic}
\end{algorithm}

\subsubsection{$\quick()$}
\label{sec:quick}
The module $\quick()$ uses additional variables
$\levelq \in \{0,1,\dots,\lmax-1\}$ and $\done \in \{\fl,\tr\}$ of group $\vcand \cap V_1$.
Each agent $v$ executes this module
only when $v.\epoch = 1$ holds.
As mentioned in section \ref{sec:status},
when an agent $v$ is assigned with status $\vcand$, it holds that $v$ is a leader and $v.\done=\fl$
or $v$ is a follower and $v.\done=\tr$.

In an execution of module $\quick()$,
each leader $v \in \vcand$ makes fair coin flips repeatedly until it sees ``tail'' for the first time
and stores on $v.\levelq$
the number of times it observes ``heads''.
Specifically, a leader with $v.\done = \fl$
makes a fair coin flip every time
it interacts with a follower (\ie an agent in $\vf$).
If the result is head
(\ie $v$ is an initiator at the interaction), it increments $\levelq$ by one (Line 36).
Otherwise, it stops coin flipping by assigning $v.\done \gets \tr$ (Line 37).
The largest $\levelq$ among all agents in $\vl$
spreads to the whole sub-population $\vcand$ via one-way epidemic.
Specifically, when two stopped agents $u,v \in \vcand$ meet,
they update their $\levelq$ to $\max(u.\levelq,v.\levelq)$
(Line 41).
When an agent $v \in \vcand$ meets an agent with larger $\levelq$
than $v.\levelq$, it becomes a follower (Line 40).
The correctness of $\quick()$ is formalized as the following lemma.

\begin{lemma}
\label{lemma:quick}
Let $\Xi = \Xi_{\plog}(\cinit{\plog},\Gamma)=C_0,C_1,\dots$.
In a configuration $C_{\unit}$,
$\Pr(|\vl|=i) < 2^{1-i}+\epsilon_i$ holds for any $i=2,3,\dots,n$
where $\epsilon_i$ is a non-negative number
such that $\sum_{i=2}^{n} \epsilon_i = O(n^{-1})$.
\end{lemma}
\begin{proof}

Coin flips in $\quick()$ are not only fair
but also independent of each other.
This is because we assume the uniformly random scheduler $\Gamma$
and at most one agent makes a coin flip at each step
(\ie at each interaction) since a coin flip is made
only when a leader and a follower meet.
Therefore, an execution of this module
correctly simulates the competition game introduced in Section
\ref{sec:keyquick} and the simulation of the game
finishes within the first $\unit$ interactions
if all of the following conditions hold in $C_{\unit}$; 
\begin{itemize}
 \item every agent $v$ is still in the first epoch,
       \ie $v.\epoch=1$ holds,
 \item $v.\levelq < \lmax$ holds for all $v \in \vcand$
 \item all agents in $\vcand$ has the same $\levelq$
       and $v.\done = \tr$ holds for all $v \in \vcand$.
\end{itemize}
Intuitively, the second condition guarantees that
no agent increases $\levelq$ to the upper limit $\lmax$
within the first $\unit$ interactions
and the third condition means that
every leader finishes coin flips and
the maximum value of $\level_q$ propagates to the whole
sub-population $\vcand$ within the first $\unit$ interactions.
The second condition $\forall v \in \vcand: v.\levelq < \lmax$
is necessary because if some agent in $\vcand$
increases $\levelq$ to $\lmax$, then it may fail to simulate
the competition game successfully.

Note that the competition game guarantees that
exactly $i$ agents survives the game with probability
at most $2^{1-i}$. Therefore, it suffices
to prove that all the three conditions hold with high probability,
\ie with probability $1-O(n^{-1})$.
Since $\cinit{\plog} \in \cstart(0)$ holds,
it directly follows from Lemma \ref{lemma:sync}
that the first condition holds with high probability.
The second condition holds with high probability
because the second condition does not hold
only when some leader
gets head $\lmax$ times in a row
and the probability that such an event happens
is at most $n(1/2)^{\lmax} \le n\cdot 2^{-5\lg n} = O(n^{-1})$.

In what follows, we prove that the third condition holds
with high probability.
In the similar way to analyze the probability of the second condition,
we can easily prove that no leader gets head $2\lg n$ times in a row
with high probability. Furthermore,
at each step, any leader meets a follower
with probability at least $|\vf|/{}_n C_2 \ge 1/n$
by Lemma \ref{lemma:status},
hence it holds
with probability $1-n\cdot n^{-2}=1-O(n^{-1})$
by Chernoff bound in the form of $\eqref{eq:lowerhalf}$
in Lemma \ref{lemma:chernoff}
that every leader meets a follower no less than $2 \lg n$ times
during the first $\smallunit (\ge \lceil 6n \lg n \rceil)$ interactions.
Therefore, with high probability,
all agents in $\vcand$ finish making coin flips
within the first $\smallunit$ interactions.
Thereafter, 
the maximum value of $\levelq$
is propagated to the whole sub-population $\vcand$
by one-way epidemic in $\vcand$.
The epidemic finishes within the next $\lfloor 8n \ln n \rfloor$ interactions with high probability by Lemma \ref{lemma:epidemic}.
Since $\smallunit + \lfloor 8n \ln n \rfloor < \unit$,
the third condition also holds with high probability.
\end{proof}

Note that an execution of $\quick()$ never eliminates all leaders from population because a leader $v$
with $v.\levelq=\max_{u \in \vcand} u.\levelq$ never becomes a follower.

\begin{algorithm}[t]
\caption{$\tourn()$}
\label{al:tourn}
\textbf{Interaction} between initiator $a_0$ and responder $a_1$:
\hspace{-1cm} \verb| |
\begin{algorithmic}[1]  
 \setcounter{ALC@line}{42}
 \IF{$i \in \{0,1\}: a_i \in \vl \wedge a_{1-i} \in \vf \wedge a_i.\ind < \rsize$}
 \STATE $a_i.\rand \gets 2a_i.\rand + i$
 \STATE $a_i.\ind \gets \max(a_i.\ind+1,\rsize)$
 \ENDIF
 \IF{$a_0, a_1 \in \vcand \wedge a_0.\ind = a_1.\ind = \rsize \wedge \exists i \in \{0,1\}: a_i.\rand < a_{1-i}.\rand$}
 \STATE $a_i.\leader \gets \fl$
 \STATE $a_i.\rand \gets a_{1-i}.\rand$
 \ENDIF
\end{algorithmic}
\end{algorithm}

\subsubsection{$\tourn()$}
\label{sec:tourn}
In the key idea depicted in Section \ref{sec:keytourn},
each leader $v$
makes fair coin flips exactly $\lceil \lg m \rceil = \Theta(\log \log n)$
times. However, this requires $\Omega(\log n \cdot \log \log n)$ states per agent
because this procedure requires not only variable $v.\rand$ that stores the results of those flips but also
variable $v.\ind$ to memorize how many times $v$ already made coin flips.
Therefore, in an execution of $\tourn()$,
each agent makes fair coin flips only $\rsize = \lceil \frac{2}{3}\lg m \rceil$ times, and we execute this module $\tourn()$ twice.
That is why we assign two epochs (\ie the second and the third epochs) to $\tourn()$.

In an execution of $\tourn()$,
each leader $v$ gets a random number, say \emph{nonce},
uniformly at random from $\{0,1,\dots,2^{\rsize}-1\}$
by making coin flips $\rsize$ times,
and stores it in $v.\rand$  (Line 43-46).
The uniform randomness of this nonce
is guaranteed because
these coin flips are not only fair
but also independent of each other,
as mentioned in Section \ref{sec:quick}.
Leaders who finishes generating a nonce begins
one-way epidemics of the largest value of these nonces
(Lines 47-50).
By Chernoff bound, 
it holds with high probability that
all leaders finishes generating its nonce
within $O(\log n)$ parallel time
and the largest value of these nonces
propagates to the whole sub-populations $\vcand$
within $O(\log n)$ parallel time.
Note that an execution of $\tourn()$ never eliminates all leaders from population because a leader 
with the largest nonce never becomes a follower.

\begin{lemma}
\label{lemma:tourn}
In an execution $\Xi = \Xi_{\plog}(\cinit{\plog},\Gamma)=C_0,C_1,\dots$,
the number of leaders become exactly one 
before some agent enters the fourth epoch (\ie $\epoch=4$)
with probability $1-O(1/\log n)$.
\end{lemma}
\begin{proof}
By Lemmas \ref{lemma:sync} and \ref{lemma:quick},
there exist at most $\loglogn$ leaders and all agents are still in the first epoch in configuration $C_{\unit}$
with probability
$1-\left(\sum_{i=\loglogn + 1}^{n} 2^{1-i}\right) - O(n^{-1}) =1-O(1/\log n)$.
Thereafter, execution $\Xi$ reaches a configuration in $\cstart(1)$
within the next $O(n \log n)$ interactions
with high probability
by Lemmas \ref{lemma:status} and \ref{lemma:timer}.

Therefore, in order to prove the lemma,
we can assume that there exists an integer $t'=O(n \log n)$
such that
there exists at most $\loglogn$ leaders in $C_{t'}$,
every agent is in the first or the second epoch
in $C_{t'}$,
and $C_{t'} \in \cstart(1)$ holds.
We say that an execution of module $\tourn()$ finishes
completely if every leader finishes generating
a nonce and
the maximum value of nonces is propagated to
the whole sub-population $\vcand$.
Since a leader generates a nonce
uniformly at random among $\{0,1,\dots,2^{\rsize}-1\}$
in each of the two executions of $\tourn()$,
the same arguments in section \ref{sec:keytourn}
yields that exactly one leader exists with probability
at least
$1-(\loglogn-1)\cdot 2^{-\rsize}\ge 1-O(\log \log n/\log^{2/3} n)$
after one execution of module $\tourn()$ finishes completely.

Each leader generates a nonce
by meeting a follower $\rsize = O(\log m)=O(\log \log n)$ times
while any leader meets a follower
with probability at least $\frac{1}{n}$ at each step
by Lemma \ref{lemma:status}.
Therefore, by Chernoff bound and Lemma \ref{lemma:epidemic},
an execution of $\tourn()$ finishes completely
within $\unit - \lfloor 4 n \ln n \rfloor \ge \lfloor 17 n \ln n \rfloor$ interactions
with high probability for sufficiently large $n$.
Hence, by Lemma \ref{lemma:sync},
both the first and the second executions of $\tourn()$
finish completely in the next $O(n \log n)$ steps
with high probability.
Therefore, by Lemma \ref{lemma:sync},
the two executions of $\tourn$ decreases the number of leaders 
from at most $\loglogn$ to exactly one
before some agent enters the fourth epoch
with probability $1-O((\log \log n/\log^{2/3}n)^2)=1-O(1/\log n)$.
%
%
\end{proof}


\begin{algorithm}[t]
\caption{$\backup()$}
\label{al:backup}
\textbf{Interaction} between initiator $a_0$ and responder $a_1$:
\hspace{-1cm} \verb| |
\begin{algorithmic}[1]  
 \setcounter{ALC@line}{50}
\IF{$a_0.\tick \wedge a_0 \in \vl \wedge a_1 \in \vf$}
\STATE $a_0.\levelb \gets \max(a_0.\levelb + 1,\lmax)$
\ENDIF
\IF{$a_0, a_1 \in \vcand \wedge \exists i \in \{0,1\}: a_i.\levelb < a_{1-i}.\levelb$}
\STATE $a_i.\levelb \gets a_{1-i}.\levelb$
\STATE $a_i.\leader \gets \fl$
\ENDIF
\STATE \ifthen{$\forall i \in \{0,1\}: a_i \in \vl$}{$a_1.\leader \gets \fl$}
\end{algorithmic}
\end{algorithm}

\subsubsection{$\backup()$}
\label{sec:backup}
This module $\backup()$ uses only one 
additional variable $\levelb$
to elect the unique leader.
For any $v \in \vcand$,
variables $v.\levelb$ is initialized by
$v.\levelb \gets 0$ 
at the first time $v.\epoch=4$ holds (Line 9).

As long as synchronization succeeds,
each $v.\tick$ is raised every $\Theta(\log n)$
(but sufficiently long) parallel time.
In an execution of $\backup()$,
each leader has a chance of making a coin flip
every time its $\tick$ is raised.
Specifically, a leader $v$ makes a coin flip
when $v$ has an interaction with a follower
and $v.\tick$ is raised at that interaction.
If $v$ sees ``head'' (\ie it is an initiator at that interaction),
it increments $v.\levelb$ by one
unless $v.\levelb$ already reaches $\lmax$ (Lines 51-53).
The largest value of $\levelb$ is propagated
via one-way epidemic in sub-population $\vcand$ (Lines 54-57).
If a leader $v$ observes the larger value of $\levelb$
than $v.\levelb$, it becomes a follower (Line 56).
Furthermore,
this module includes simple leader election
\cite{original};
when two leaders interact and they observe that
they have the same value of $\levelb$ at Line 58,
then the responder becomes a follower.
Note that an execution of $\backup()$
never eliminates all leaders from population because a leader
with the largest value of $\levelb$ never becomes a follower.

\begin{lemma}
\label{lemma:fourth}
Let $C$ be any configuration in $\call(\plog)$
and let $\Xi = \Xi_{\plog}(C,\Gamma)=C_0,C_1,\dots$.
Execution $\Xi$ reaches a configuration where 
all agents are in the fourth epoch
within $O(\log n)$ parallel time with high probability
and in expectation.
\end{lemma}
\begin{proof}
After the first step of $\Xi$, at least one agent $v$ in $\vtimer$
always exists. By Lemma \ref{lemma:timer}, $v$ enters the forth epoch
within $O(\log n)$ parallel time with high probability
and in expectation.
Since the largest value of epoch is propagated to the whole population 
via one-way epidemic, all agents enters the fourth epoch
within $O(\log n)$ parallel time with high probability
and in expectation
by Lemma \ref{lemma:epidemic}.
\end{proof}

\begin{lemma}
\label{lemma:truebackup} 
Let $C$ be any configuration where
all agents are in the fourth epoch
and let $\Xi = \Xi_{\plog}(C,\Gamma)=C_0,C_1,\dots$.
Then $\Xi$ reaches a configuration where
there exists exactly one leader
within $O(n)$ parallel time in expectation.
\end{lemma}
\begin{proof}
Execution $\Xi$ elects the unique leader
within $O(n)$ parallel time in expectation
because module $\backup()$ includes the simple leader election
mechanism \cite{original}, \ie one leader becomes a follower
when two leaders meet. 
\end{proof}

\begin{definition}
We define $\bstart$
as the set of all configurations in $\ccolor(0)$
where every agent is in the fourth epoch
(\ie $\epoch = 4$),
and $v.\levelb \le 1$
holds for all agents $v \in \vcand$.
\end{definition}

\begin{lemma}
\label{lemma:bstart} 
Execution $\Xi = \Xi_{\plog}(\cinit{\plog},\Gamma)$
reaches a configuration in $\bstart$
within $O(\log n)$ parallel time
with high probability. 
\end{lemma}
\begin{proof}
This lemma directly follows from
Lemma \ref{lemma:sync}
and the definition of $\backup()$.
\end{proof}

\begin{lemma}
\label{lemma:backup} 
Let $C$ be any configuration in $\bstart$
and let $\Xi = \Xi_{\plog}(C),\Gamma)=C_0,C_1,\dots$.
Then $\Xi$ reaches a configuration where
there exists exactly one leader
within $O(\log^2 n)$ parallel time
in expectation
\end{lemma}
\begin{proof}

By applying Lemma \ref{lemma:sync} repeatedly,
it holds for sufficiently large $O(\log^2 n)$ parallel time
with probability $1-O(\log n/ n)$ that
synchronization does not fail and
no agent raises $\tick$ twice within
any $\unit-\lceil 4n \ln n \rceil \ge \lceil 17 \ln n \rceil$ steps.
Thus, in the following, we assume that no agent raises $tick$ twice within any $\lceil 17 \ln n \rceil$ steps.


Define $B = \max_{v \in \vcand \cap V_4} v.\levelb$.
By Lemma \ref{lemma:sync},
each leader raises its $\tick$
in every $O(\log n)$ parallel time with high probability
and makes a coin flip if it meets a follower
at that interaction.
Thus, 
each leader increments its $\levelb$ with probability $1/4$
in every $O(\log n)$ parallel time
because $|\vf| \ge n/2$ by Lemma \ref{lemma:status}.
Therefore, the maximum value $B$ is increased by one
within $O(\log n)$ parallel time in expectation,
thus $B$ reaches $\lmax$ within $O(\log^2 n)$ parallel time
in expectation.

Consider that now $B$ is increased from $k$ to $k+1$.
At this time,
only one leader has the largest value $B=k+1$ in $\levelb$.
Thereafter, this value $k+1$ is propagated to the whole
sub-population $\vcand$ within $\lceil 8n \ln n \rceil$ steps
with high probability by Lemma \ref{lemma:epidemic},
during which no leader makes coin flips twice
by the above assumption.
Therefore, the number of leaders decreases almost by half,
specifically decreases from $1+i$
to at most $1+\lfloor i/2\rfloor$,
with probability $1/2-O(1/n)$.
Clearly, in execution $\Xi$,
$B$ is eventually 
increased from $0$ or $1$ to $\lmax$.
Therefore, ignoring the probability that synchronization fails
or the one-way epidemic does not finish 
within $\lfloor 8n \ln n \rfloor$ steps,
we can observe by Chernoff bound that the number of leaders becomes one before $B$ reaches $\lmax$
with probability $1-O(\log n /n)$.

Even if $B$ reaches $\lmax$
before one leader is elected,
%
$\Xi$ elects the unique leader
within $O(n)$ parallel time in expectation thereafter
by Lemma \ref{lemma:truebackup}.
Therefore, $\Xi$ reaches a configuration where exactly one leader exists within
$O(\log^2 n) + O(\log n / n) \cdot O(n) = O(\log^2 n)$ in expectation.
\end{proof}

\begin{theorem}
\label{theorem:expectation}
 Let $\Xi = \Xi_{\plog}(\cinit{\plog},\Gamma)=C_0,C_1,\dots$.
Execution $\Xi$ reaches a configuration where 
exactly one leader exists within $O(\log n)$ parallel time
in expectation.
\end{theorem}
\begin{proof}
First, Lemmas \ref{lemma:tourn} and \ref{lemma:fourth},
execution $\Xi$ reaches a configuration where
exactly one leader exists within $O(\log n)$ parallel time
with probability $1-O(1/\log n)$.
Second, by Lemma \ref{lemma:bstart}, 
execution $\Xi$ reaches a configuration in $\bstart$
within $O(\log n)$ parallel time with high probability.
Thereafter, execution $\Xi$ reaches a configuration where 
exactly one leader exists within $O(\log^2 n)$ parallel time
in expectation by Lemma \ref{lemma:backup}.
Finally, Lemmas \ref{lemma:fourth} and \ref{lemma:truebackup}
shows that starting from any configuration in $\call(\plog)$,
$\Xi$ reaches a configuration where exactly one leader is
elected within $O(n)$ parallel time in expectation.
To conclude, starting from initial configuration $\cinit{\plog}$,
execution $\Xi$ reaches a configuration where the unique leader is elected within
$O(\log n)+ O(1/\log n)\cdot O(\log^2 n)+O(1/n)\cdot O(n)=O(\log n)$ parallel time in expectation.
\end{proof}

\section{Discussion towards Symmetric Transitions}
\label{sec:symmetric}
In the field of PP model,
several works are devoted to design
a \emph{symmetric protocol}.
Suppose that two agents have an interaction
and their states changes from $p,q$ to $p',q'$, respectively.
A protocol is symmetric if
$p=q \Rightarrow p'=q'$ always hold.
In other words, a symmetric protocol
is a protocol that does not utilize the roles of the two agents
at an interaction, initiator and responder.
This property is important for some applications such as
chemical reaction networks.

The proposed protocol $\plog$ described above is not symmetric,
however, we can make it symmetric 
by the following strategy.
Protocol $\plog$ performs asymmetric actions 
only for assignment of status
(Section \ref{sec:status})
and flipping fair and independent coins
(Sections \ref{sec:quick}, \ref{sec:tourn}, and \ref{sec:backup}).
To assign the agents their statuses
by symmetric transitions,
we only have to add additional status $Y$ 
and make the following three rules:
$X \times X \to Y \times Y$, \hspace{5pt}
$Y \times Y \to X \times X$, \hspace{5pt}
$X \times Y \to A \times B$. 
Furthermore, similarly to the original rules of $\plog$,
when an agent $v$ with status $X$ or $Y$ meets
an agent with status $A$ or $B$,
$v$ gets status $A$ but it becomes a follower.
This modification does not make any harmful influence
on the analysis of stabilization time,
at least asymptotically.
Coin flips are dealt with in the same way.
We assign a \emph{coin status}
$J$, $K$, $F_0$, or $F_1$ to each follower.
Every time a leader $v$ becomes a follower,
initial status $J$ is assigned to $v$.
Thereafter, when two followers meet,
they change their coin statuses
according to the following rules:
$J \times J \to K \times K$, \hspace{5pt}
$K \times K \to J \times J$, \hspace{5pt}
$J \times K \to F_0 \times F_1$. 
These rules guarantees that
the numbers of the followers with state $F_0$
and $F_1$ are always equal.
Therefore, a leader can make a fair and independent
coin flip every time it meets a follower
whose coin state is $F_0$ or $F_1$.
If it meets a follower with coin state $F_0$ (resp.~$F_1$),
it recognizes that the result of the flip is head (resp.~tail).

\section{Conclusion}
In this paper, we gave
a leader election protocol
with logarithmic stabilization time 
and with logarithmic number of agent stats in the 
population protocol model.
Given a rough knowledge $m$ of the population size $n$
such that $m \ge \log_2 n$ and $m=O(\log n)$,
the proposed protocol guarantees that
exactly one leader is elected from $n$ agents within $O(\log n)$ parallel time in expectation.

\printbibliography
\end{document}